\documentclass[11pt,twoside]{article}
\usepackage{amsmath,amssymb,latexsym,graphicx,hyperref,epstopdf,subcaption,color}
\usepackage{algpseudocode,algorithm,epstopdf}
\usepackage[capitalise]{cleveref}
\numberwithin{figure}{section}
\normalsize
\newtheorem{theorem}{Theorem}
\newtheorem{lemma}{Lemma}
\newtheorem{corollary}{Corollary}

\newenvironment{proof}{{\sc Proof. }}{\hfill$\Box$\vspace{0.2in}}
\usepackage{chngcntr}
\counterwithout{figure}{section}

\title{Improved approximation algorithms for path vertex covers in regular graphs}%
\author{An~Zhang\thanks{Department of Mathematics, Hangzhou Dianzi University.  Hangzhou 350018, China.
	\texttt{Emails: \{anzhang,chenyong\}@hdu.edu.cn}}
	\and
	Yong~Chen${^*}$
	\and
	Zhi-Zhong~Chen\thanks{Division of Information System Design, Tokyo Denki University.  Saitama 350-0394, Japan.
	\texttt{Email: zzchen@dendai.ac.jp}}
	\and
	Guohui~Lin\thanks{Department of Computing Science, University of Alberta. Edmonton, Alberta T6G 2E8, Canada.
	\texttt{Email: guohui@ualberta.ca}}
	\thanks{Correspondence author.}}%
\date{\today}
\begin{document}
\maketitle

\begin{abstract}
Given a simple graph $G = (V, E)$ and a constant integer $k \ge 2$,
the $k$-path vertex cover problem ({\sc P$k$VC}) asks for a minimum subset $F \subseteq V$ of vertices such that
the induced subgraph $G[V - F]$ does not contain any path of order $k$.
When $k = 2$, this turns out to be the classic vertex cover ({\sc VC}) problem, which admits a $\left(2 - {\rm \Theta}\left(\frac 1{\log|V|}\right)\right)$-approximation.
The general {\sc P$k$VC} admits a trivial $k$-approximation;
when $k = 3$ and $k = 4$, the best known approximation results for {\sc P$3$VC} and {\sc P$4$VC} are a $2$-approximation and a $3$-approximation,
respectively.
On $d$-regular graphs, the approximation ratios can be reduced to
$\min\left\{2 - \frac 5{d+3} + \epsilon, 2 - \frac {(2 - o(1))\log\log d}{\log d}\right\}$ for {\sc VC} ({\it i.e.}, {\sc P$2$VC}),
$2 - \frac 1d + \frac {4d - 2}{3d |V|}$ for {\sc P$3$VC},
$\frac {\lfloor d/2\rfloor (2d - 2)}{(\lfloor d/2\rfloor + 1) (d - 2)}$ for {\sc P$4$VC},
and $\frac {2d - k + 2}{d - k + 2}$ for {\sc P$k$VC} when $1 \le k-2 < d \le 2(k-2)$.
By utilizing an existing algorithm for graph defective coloring, we first present a
$\frac {\lfloor d/2\rfloor (2d - k + 2)}{(\lfloor d/2\rfloor + 1) (d - k + 2)}$-approximation for {\sc P$k$VC} on $d$-regular graphs when $1 \le k - 2 < d$.
This beats all the best known approximation results for {\sc P$k$VC} on $d$-regular graphs for $k \ge 3$,
except for {\sc P$4$VC} it ties with the best prior work and in particular they tie at $2$ on cubic graphs and $4$-regular graphs.
We then propose a $1.875$-approximation and a $1.852$-approximation for {\sc P$4$VC} on cubic graphs and $4$-regular graphs, respectively.
We also present a better approximation algorithm for {\sc P$4$VC} on $d$-regular bipartite graphs.

\paragraph{Keywords:}
Path vertex cover; regular graph; defective coloring; maximum independent set; approximation algorithm 
\end{abstract}

\section{Introduction}
We investigate a vertex deletion problem called the {\em minimum $k$-path vertex cover} problem, denoted as {\sc P$k$VC},
which is a generalization of the classic minimum vertex cover ({\sc VC}) problem~\cite{GJ79}.
The {\sc P$k$VC} problem has been studied for more than three decades in the literature,
and it has applications in wireless sensor networks such as constructing optimal connectivity paths and
in networking security such as monitoring the message traffic and detecting malicious attack~\cite{ACB12}.

Given a simple graph $G = (V, E)$ and a constant integer $k \ge 2$, a $k$-path (or, a path of order $k$) is a simple path containing $k$ vertices;
the {\sc P$k$VC} problem asks for a minimum subset $F \subseteq V$ of vertices such that
the induced subgraph $G[V - F]$ (set minus operation) does not contain any $k$-path~\cite{BKK11,JT13,Jak15}.
When $k = 2$, this turns out to be {\sc VC}.
In the literature, a $k$-path vertex cover is also called a vertex $k$-path cover~\cite{BJK13},
or a vertex cover $P_k$~\cite{TZ11,TY13},
or a $P_k$ vertex cover~\cite{DMM15},
or a $k$-observer~\cite{ACB12,RSU15}.
Also, when $F \subseteq V$ is a ($2$-path) vertex cover, $V - F$ is an independent set,
and when $F \subseteq V$ is a $3$-path vertex cover, $V - F$ is a dissociation set.
The maximum independent set ({\sc MIS}) problem is another classic NP-hard problem~\cite{GJ79};
the maximum dissociation set problem is also classic, was introduced more than three decades ago by Yannakakis~\cite{Yan81}, and is NP-hard even on bipartite graphs.

The concept of $k$-path vertex covers, and many related ones, form a line of research in graph theory.
The minimum cardinality of a $k$-path vertex cover, for $k \ge 2$, in the graph $G = (V, E)$ is denoted by $\psi_k(G)$~\cite{BJK13}.
Clearly, $\psi_2(G) = |V| - \alpha(G)$, where $\alpha(G)$ is the independent number, that is, the maximum cardinality of an independent set in $G$.
The maximum cardinality of a dissociation set in $G$, also known as the $1$-dependence number~\cite{Fav85,Fav88}, is denoted as $diss(G)$~\cite{ODF11},
and we have $\psi_3(G) = |V| - diss(G)$.
Let $P_n$, $C_n$ and $K_n$ denote a simple path, a simple cycle and a complete graph with $n$ vertices, respectively,
then $\psi_k(P_n) = \lfloor\frac nk\rfloor$, $\psi_k(C_n) = \lceil \frac nk\rceil$ and $\psi_k(K_n) = n-k+1$~\cite{ACB12,BJK13}.
When $G$ is some special graph~\cite{BKK11,BJK13,JT13,BKS14,Jak15}, the exact value of $\psi_k(G)$, for some values of $k$, can also be computed in polynomial time.
For the cases where the exact value of $\psi_k(G)$ is unable to be computed in polynomial time,
there are works proving several lower and/or upper bounds on $\psi_k(G)$;
to name a few, Bre{\v s}ar et al.~\cite{BKK11} proved that $\psi_3(G) \le \frac {2|V|+|E|}6$ and $\psi_k(G) \le |V| - \frac {k-1}k \sum_{v\in V} \frac 2{1+d(v)}$,
where $d(v)$ is the degree of $v$ in $G$;
Bre{\v s}ar et al.~\cite{BJK13} showed that $\psi_k(G) \ge \frac {d-k+2}{2d-k+2} |V|$ for $d$-regular graphs when $d \ge k-1$.

The {\sc P$k$VC} problem is NP-hard for every $k \ge 2$~\cite{BKK11,ACB12}.
From the inapproximability (hardness of approximation) perspective,
the {\sc VC} ({\it i.e.}, {\sc P$2$VC}) problem is APX-complete even on cubic graphs~\cite{AK00};
it cannot be approximated within $10\sqrt{5} - 21 \approx 1.3606$ unless P = NP~\cite{DS05} and
it cannot be approximated within any constant factor less than $2$~\cite{KR08} under the unique game conjecture~\cite{Kho02}.
%
Bre{\v s}ar et al.~\cite{BKK11} proved that for any $\rho \ge 1$,
a $\rho$-approximation for the {\sc P$k$VC} problem implies a $\rho$-approximation for the {\sc VC} problem.
It follows that it is NP-hard to approximate the {\sc P$k$VC} problem, for every $k \ge 3$, within $1.3606$ too, unless P = NP.
%
Recall that the maximum dissociation set problem is NP-hard even on (some sub-classes of) planar bipartite graphs~\cite{Yan81,BCL04,ODF11};
the {\sc P$3$VC} problem is shown NP-hard on cubic planar graphs with girth $3$~\cite{TY13}.
The {\sc P$4$VC} problem is proven APX-complete on cubic bipartite graphs and on $K_{1,4}$-free graphs~\cite{DMM15}.

From the approximation algorithm perspective,
the simple greedy algorithm, that iteratively takes all the $k$ vertices from a $k$-path in the remaining graph until there is no $k$-path left,
is a $k$-approximation for the {\sc P$k$VC} problem, for every $k \ge 2$.
Furthermore, the VC ({\it i.e.}, {\sc P$2$VC}) problem admits a $\left(2 - {\rm \Theta}\left(\frac 1{\log |V|}\right)\right)$-approximation~\cite{Kar09};
the {\sc P$3$VC} problem admits a primal-dual $2$-approximation~\cite{TZ11};
the {\sc P$4$VC} problem admits a primal-dual $3$-approximation~\cite{CCC14}.

On $d$-regular graphs, where $d \ge 3$, the {\sc P$k$VC} problem can be approximated better, for every $k \ge 2$.
The approximation ratio for the {\sc VC} ({\it i.e.}, {\sc P$2$VC}) problem can be reduced to
$\min\Big\{2 - \frac 5{d+3} + \epsilon$~\cite{BF95}, $2 - \frac {(2 - o(1))\log\log d}{\log d}$~\cite{Hal00}$\Big\}$;
Ries et al.~\cite{RSU15} gave a $\left(2 - \frac 1d + \frac {4d-2}{3d|V|}\right)$-approximation for {\sc P$3$VC}
and Devi et al.~\cite{DMM15} presented a greedy $\frac {\lfloor d/2\rfloor (2d - 2)}{(\lfloor d/2\rfloor + 1) (d - 2)}$-approximation for {\sc P$4$VC}
(using the lower bound $\frac {d-k+2}{2d-k+2} |V|$ given by Bre{\v s}ar et al.~\cite{BJK13}).
Ries et al.~\cite{RSU15} also proposed a $\frac {2d - k + 2}{d - k + 2}$-approximation for {\sc P$k$VC} when $1 \le k-2 < d \le 2(k-2)$.

More specifically on cubic graphs, Tu and Yang \cite{TY13} gave a $\frac {25}{16} (= 1.5625)$-approximation for the {\sc P$3$VC} problem;
Ries et al.~\cite{RSU15} claimed that their approximation algorithm for the {\sc P$3$VC} problem can be reduced to $1.389 + o(1)$.
Table~\ref{tab01} summarizes the best approximation ratios prior this work for the {\sc P$k$VC} problem, for $k \ge 2$.

\begin{table}[ht]
\caption{The best prior approximation ratios for the {\sc P$k$VC} problem.\label{tab01}}
\begin{center}
\begin{tabular}{c||c||c|c|c}
{\sc P$k$VC}	& General graphs
					& Cubic 					& $4$-regular 					& $d$-regular ($d \ge 5$)\\
\hline
\hline
$k=2$ 			& $2 - {\rm \Theta}\left(\frac 1{\log |V|}\right)$~\cite{Kar09}
					& \multicolumn{3}{c}{$\min\Big\{2 - \frac 5{d+3} + \epsilon$~\cite{BF95}, $2 - \frac {(2 - o(1))\log\log d}{\log d}$~\cite{Hal00}$\Big\}$}\\[4pt]
\hline
$k=3$ 			& $2$~\cite{TZ11}
					& $1.389+o(1)$~\cite{RSU15}	& $\frac 74 +o(1)$~\cite{RSU15} & $2 - \frac 1d + \frac {4d-2}{3d|V|}$~\cite{RSU15}\\[4pt]
\hline
$k=4$ 			& $3$~\cite{CCC14}
					& $2$~\cite{DMM15}		& $2$~\cite{DMM15}			& $\frac {\lfloor d/2\rfloor (2d - 2)}{(\lfloor d/2\rfloor + 1) (d - 2)}$~\cite{DMM15}\\[4pt]
\hline
$k \ge 5$ 		& $k$
					& n.a.			& \multicolumn{2}{c}{$\frac {2d - k + 2}{d - k + 2}$, when $k-2 < d \le 2(k-2)$~\cite{RSU15}}\\
\end{tabular}
\end{center}
\end{table}

In this paper, we aim to design improved approximation algorithms for the {\sc P$k$VC} problem, for $k \ge 3$, on $d$-regular graphs.
To this purpose, in Section 2 we first employ an existing polynomial time graph defective coloring algorithm
to design a simple yet effective approximation algorithm for the {\sc P$k$VC} problem,
and we are able to show that its approximation ratio is $\frac {\lfloor d/2\rfloor (2d - k + 2)}{(\lfloor d/2\rfloor + 1) (d - k + 2)}$, when $1 \le k - 2 < d$.
This beats all the best prior approximation results except for {\sc P$4$VC},
where it ties with the approximation by Devi et al.~\cite{DMM15} and in particular at $2$ on cubic graphs and on $4$-regular graphs.
In Section 3, we first prove a lower bound on $\psi_4(G)$ when $G$ is $d$-regular,
then integrate the graph defective coloring algorithm and the current best approximation algorithm for the {\sc MIS} problem on degree-bounded graphs,
to design a $\frac {15}8$-approximation for {\sc P$4$VC} on cubic graphs.
When $d \ge 4$ is even, we show in Section 4 how to compute $\frac d2$ $4$-path vertex covers in $G$ and
selecting the minimum one gives a $\frac {(3d-2)(2d-2)}{(3d+4)(d-2)}$-approximation for {\sc P$4$VC}.
This turns out to be a $\frac {15}8$-approximation for {\sc P$4$VC} on $4$-regular graphs.
Also in Section 4, we are able to provide a better analysis to show that the algorithm is actually a $1.852$-approximation on $4$-regular graphs,
and construct an instance to show that the ratio $1.852$ is almost tight for {\sc Approx2} $4$-regular graphs.
Lastly, in Section 5, we propose a $\frac {d^2}{d^2 - d + 1}$-approximation algorithm for {\sc P$4$VC} on $d$-regular bipartite graphs.
We conclude the paper in Section 6 with some remarks.

\section{An approximation for {\sc P$k$VC} on $d$-regular graphs}
We consider the {\sc P$k$VC} problem for $k \ge 3$, for which the best approximation algorithms prior our work are summarized in Table~\ref{tab01}.

As a consequence of Lov{\' a}sz's graph decomposition~\cite{Lov66},
Cowen and Jesurum~\cite{CGJ97} have proved the following result on {\em defective coloring} for any graph of maximum degree $\Delta$,
where a {\em defective} $(p, q)$-coloring colors the vertices of the graph using $p$ colors
such that each vertex is adjacent to at most $q$ the same colored neighbors.
A $(p, 0)$-coloring is the classic vertex coloring using $p$ colors.

\begin{theorem}{\rm \cite{CGJ97}}
\label{thm01}
Any graph $G = (V, E)$ of maximum degree $\Delta$ can be $\left(p, \left\lfloor \frac {\Delta}p \right\rfloor\right)$-colored in $O(\Delta |E|)$ time.
\end{theorem}

Let $G = (V, E)$ be a $d$-regular graph.
Using Theorem~\ref{thm01} by setting $p = \lfloor \frac d2 \rfloor + 1$, we have a defective $(p, 1)$-coloring for the graph $G$.
In this defective $(p, 1)$-coloring,
suppose these $p$ colors are $1, 2, \ldots, p$;
let $V^i$ denote the subset of the vertices colored $i$.
Then clearly,
\begin{enumerate}
\parskip=0pt
\item[1)]
	$\{V^1, V^2, \ldots, V^p\}$ is a partition of the vertex set $V$, and
\item[2)]
	the subgraph induced on $V^i$, $G[V^i]$, does not contain any $3$-path,
	suggesting $F^i = V - V^i$ is a $3$-path vertex cover (and thus it is also a $k$-path vertex cover for any $k \ge 3$), for every $i = 1, 2, \ldots, p$.
\end{enumerate}

It follows that the minimum among $F^1, F^2, \ldots, F^p$, denoted as $F^{\min}$, has size
\begin{equation}
\label{eq1}
|F^{\min}| \le \left(1 - \frac 1p\right) |V|.
\end{equation}

Recall that when $d \ge k-1$, Bre{\v s}ar et al.~\cite{BJK13} have proved the following lower bound on $\psi_k(G)$:
\begin{equation}
\label{eq2}
\psi_k(G) \ge \frac {d-k+2}{2d-k+2} |V|.
\end{equation}
Therefore, $F^{\min}$ turns out to be an approximate solution within ratio
\[
\frac {|F^{\min}|}{\psi_k(G)} \le \frac {(p - 1) (2d - k + 2)}{p (d - k + 2)} = \frac {\lfloor d/2 \rfloor (2d - k + 2)}{(\lfloor d/2 \rfloor + 1) (d - k + 2)}.
\]

That is, using the existing graph defective coloring algorithm,
we can design an algorithm, denoted as {\sc DC}, to first compute in $O(d^2 |V|)$ time a defective $(\lfloor \frac d2 \rfloor + 1, 1)$-coloring
for the input $d$-regular graph $G = (V, E)$,
then in $O(d)$ time find the color $i$ with the most vertices,
and last return $F$ containing all the vertices not colored $i$.
See Figure~\ref{fig01} for a high-level description of the algorithm {\sc DC}.
We conclude with Theorem~\ref{thm02}.

\begin{theorem}
\label{thm02}
The algorithm {\sc DC} for the {\sc P$k$VC} problem on $d$-regular graphs, where $1 \le k - 2 < d$,
is an $O(d^2 |V|)$-time $\frac {\lfloor d/2 \rfloor (2d - k + 2)}{(\lfloor d/2 \rfloor + 1) (d - k + 2)}$-approximation.
\end{theorem}

\begin{figure}[ht]
\begin{center}
\framebox{
\begin{minipage}{4.5in}
	The algorithm {\sc DC} for {\sc P$k$VC} on $d$-regular graphs:
	\begin{description}
	\item[Step 1.]
		Set $p = \lfloor \frac d2 \rfloor + 1$ and compute a defective $(p, 1)$-coloring;
	\item[Step 2.]
		let $i$ denote the color with the most vertices,\\
		and $V^i$ denote the set of vertices colored $i$;
	\item[\hspace{0.15in}]
		return $F = V - V^i$.
	\end{description}
\end{minipage}}
\end{center}
\caption{The approximation algorithm {\sc DC} for {\sc P$k$VC} on $d$-regular graphs.\label{fig01}}
\end{figure}

On $d$-regular graphs, our algorithm {\sc DC} for the {\sc P$k$VC} problem for $k \ge 3$ beats all the best prior approximation results
except for $k = 4$ ours ties with the approximation by Devi et al.~\cite{DMM15}.
Table~\ref{tab02} summarizes the improvement over the corresponding entries in Table~\ref{tab01}.
From this table, we see that currently there is no published approximation result for {\sc P$k$VC} on $d$-regular graphs such that $d \le k-2$ (and $k \ge 5$).
For {\sc P$4$VC}, the approximation ratios are constants strictly less than $2$ for all $d \ge 5$, while they are $2$ for both $d = 3, 4$.
In the next two sections, we design a $\frac {15}8$-approximation for {\sc P$4$VC} on cubic graphs and
a $1.852$-approximation for {\sc P$4$VC} on $4$-regular graphs, respectively.

\begin{table}[ht]
\caption{The updated approximation ratios for {\sc P$k$VC}, $k \ge 3$, on regular graphs.\label{tab02}}
\begin{center}
\begin{tabular}{c||c|c|c}
{\sc P$k$VC}	& Cubic 					& $4$-regular 					& $d$-regular ($d \ge 5$)\\
\hline
\hline
$k=3$ 			& $\frac 54$~(Theorem~\ref{thm02})	& $\frac {14}9$~(Theorem~\ref{thm02}) & $\frac {\lfloor d/2 \rfloor (2d - 1)}{(\lfloor d/2 \rfloor + 1) (d - 1)}$~(Theorem~\ref{thm02})\\[4pt]
\hline
$k=4$ 			& $2$~\cite{DMM15}		& $2$~\cite{DMM15}			& $\frac {\lfloor d/2\rfloor (2d - 2)}{(\lfloor d/2\rfloor + 1) (d - 2)}$~\cite{DMM15}\\[4pt]
\hline
$k \ge 5$ 		& n.a.			& \multicolumn{2}{c}{$\frac {\lfloor d/2 \rfloor (2d - k + 2)}{(\lfloor d/2 \rfloor + 1) (d - k + 2)}$, when $k-2 < d$~(Theorem~\ref{thm02})}\\
\end{tabular}
\end{center}
\end{table}

\section{{\sc P$4$VC} on cubic graphs}
In this and the next sections, we will design improved approximation algorithms for {\sc P$4$VC} on cubic graphs and on $4$-regular graphs,
with performance ratios $\frac {15}8$ and $1.852$, respectively.
With them, all the approximation ratios for {\sc P$4$VC} on $d$-regular graphs become strictly less than $2$.

Let $G = (V, E)$ denote the input $d$-regular graph (we set $d = 3$ later after we develop a lower bound on general $d$).
We first examine some structural properties associated with the optimal $4$-path vertex covers in $G$.

\begin{lemma}
\label{lemma01}
Let $G = (V, E)$ be a $d$-regular graph.
Then, $\psi_4(G) \ge \frac {d-1}{2d} |V| - \frac 1{d^2} f$,
where $f$ denotes the total number of vertices on the $3$-cycles in $G[V - F^*]$ and $F^*$ is an optimal $4$-path vertex cover in $G$.
\end{lemma}
\begin{proof}
Let $F^*$ be an optimal $4$-path vertex cover in $G$;
then the subgraph $G[V - F^*]$ does not contain any $4$-path.
In other words, all the connected components of $G[V - F^*]$ can be classified into the following five kinds:
\begin{enumerate}
\parskip=0pt
\item[A)]
	a $1$-path (also called a {\em singleton}),
\item[B)]
	a $2$-path,
\item[C)]
	a $3$-path,
\item[E)]
	a $K_{1,\ell}$ claw, for $\ell = 3, 4, \ldots, d$, and
\item[F)]
	a $C_3$ cycle (also called a {\em triangle});
\end{enumerate}
let $a, b, c, e, f$ denote the total numbers of vertices in these five kinds of components, respectively.
It follows that
\begin{equation}
\label{eq3}
|V| = |F^*| + a + b + c + e + f.
\end{equation}

Because $G$ is $d$-regular, the number of edges connecting a vertex of $F^*$ and a vertex of $V - F^*$ is at least
\[
d a + (d-1) b + \left(d - \frac 43\right) c + \left(\min_{3 \le \ell \le d} \left\{d - 2 + \frac 2{\ell + 1}\right\}\right) e + (d-2) f.
\]
Given that each vertex of $F^*$ can be incident with at most $d$ such edges,
and that $\min\limits_{3 \le \ell \le d} \left\{d - 2 + \frac 2{\ell + 1}\right\}$ achieves at $\ell = d$, we have
\[
d |F^*| \ge \left(d - 2 + \frac 2{d+1}\right) (a + b + c + e) + (d-2) f = \frac {d (d-1)}{d+1} (|V| - |F^*|) - \frac 2{d+1} f.
\]
It follows from Eq.~(\ref{eq3}) that, when $G$ is $d$-regular,
\begin{equation}
\label{eq4}
\psi_4(G) = |F^*| \ge \frac {d-1}{2d} |V| - \frac 1{d^2} f.
\end{equation}
This proves the lemma.
\end{proof}

We now consider only $d = 3$, that is, $G$ is cubic.

In this case, Lemma~\ref{lemma01} (or Eq.~(\ref{eq4})) states that $\psi_4(G) \ge \frac 13 |V| - \frac 19 f$.
Therefore, one sees that when the number of triangles in $G[V - F^*]$ is small,
we can expect this new lower bound to be more effective, compared against the lower bound $\frac 14 |V|$ stated in Eq.~(\ref{eq2}).
For example, when $f \le \frac 35 |V|$, we have
\[
\psi_4(G) \ge \frac 19 \left(3 - \frac 35\right) |V| = \frac 4{15} |V|.
\]
It follows that the defective $(2, 1)$-coloring for $G$ gives a $4$-path vertex cover $F^{\min}$, see Eq.~(\ref{eq1}), satisfying
\begin{equation}
\label{eq5}
\frac {|F^{\min}|} {\psi_4(G)} \le \frac 12 \times \frac {15}4 = \frac {15}8.
\end{equation}

On the other hand, when $f > \frac 35 |V|$, that is, there are a considerable number of triangle components in $G[V - F^*]$,
we will construct a new graph denoted as ${\cal G} = (T, R)$ from the input graph $G = (V, E)$ as follows.
For every triangle in $G$, we create a distinct vertex of $T$ in ${\cal G}$;
two vertices of $T$ are adjacent if and only if the corresponding two triangles of $G$ share a common edge or they are connected by an edge in $G$.%
\footnote{We remark that two distinct triangles of the cubic graph $G$ either share exactly one edge, or have no vertex in common.}
One may easily verify that the graph ${\cal G} = (T, R)$ can be constructed in $O(|V|)$ time, $|T| \in O(|V|)$,
and the maximum degree of the vertices of ${\cal G}$ is at most $3$.
Moreover, a subgraph of $G$ that is a collection of triangle components one-to-one corresponds to an independent set of ${\cal G}$.

Recall the best approximation algorithm for the {\sc MIS} problem on degree-$\Delta$ graphs by Berman and Fujito~\cite{BF95}
has a performance ratio $\frac {\Delta + 3}5 + \epsilon$ for any small positive $\epsilon$.
We next run this approximation algorithm on ${\cal G}$ to obtain an independent set $I$ of ${\cal G}$,
and therefore (roughly, by ignoring $\epsilon$)
\[
|I| \ge \frac 56 \alpha({\cal G}) \ge \frac 5{18} f,
\]
where $\alpha({\cal G})$ is the independence number of ${\cal G}$, which corresponds to the maximum number of non-adjacent triangles in $G$.
Let $F$ denote the set of the vertices of $V$ not on the triangles of $I$;
then $F$ is a solution to the {\sc P$4$VC} problem on $G$, and its cardinality is
\begin{equation}
\label{eq6}
|F| \le |V| - \frac 56 f = \frac 58 |V| + \frac 38 |V| - \frac 56 f \le \frac 58 |V| - \frac 5{24} f,
\end{equation}
where the last inequality is due to $f > \frac 35 |V|$.

Combining Eq.~(\ref{eq6}) and the lower bound in Eq.~(\ref{eq4}), we have
\begin{equation}
\label{eq7}
\frac {|F|} {\psi_4(G)} \le \frac 58 \times 3 = \frac {15}8.
\end{equation}

From Eqs.~(\ref{eq5}) and (\ref{eq7}),
we can design an algorithm, denoted as {\sc Approx1}, to first compute in $O(|V|)$ time a defective $(2, 1)$-coloring for the input cubic graph $G = (V, E)$,
then in $O(1)$ time find the color with less vertices and set $F^{\min}$ to contain all these vertices.
It also constructs the triangle graph ${\cal G} = (T, R)$ from $G = (V, E)$ and
applies the best approximation algorithm for {\sc MIS} to compute an independent set $I$ in ${\cal G}$,
then it sets $F$ to contain all the vertices of $V$ not on the triangles of $I$.
Lastly, it returns the smaller one between $F^{\min}$ and $F$ as the final solution.
See Figure~\ref{fig02} for a high-level description of the algorithm {\sc Approx1},
of which the running time is dominated by the running time of the best approximation algorithm for the {\sc MIS} problem on degree-$3$ graphs.
\begin{figure}[ht]
\begin{center}
\framebox{
\begin{minipage}{5.0in}
	The algorithm {\sc Approx1} for {\sc P$4$VC} on cubic graphs:
	\begin{description}
	\item[Step 1.]
		Compute a defective $(2, 1)$-coloring for the input graph $G$;
	\item[Step 2.]
		find the color with less vertices and set $F^{\min}$ to contain them;
	\item[Step 3.]
		construct the triangle graph ${\cal G}$ from $G$;
	\item[Step 4.]
		call the $\frac 65$-approximation to compute an independent set $I$ in ${\cal G}$,\\
		and set $F$ to contain all the vertices not on the triangles of $I$;
	\item[Step 5.]
		return the smaller one between $F^{\min}$ and $F$.
	\end{description}
\end{minipage}}
\end{center}
\caption{A high-level description of the approximation algorithm {\sc Approx1} for {\sc P$4$VC} on cubic graphs.\label{fig02}}
\end{figure}
We thus conclude with Theorem~\ref{thm03}.

\begin{theorem}
\label{thm03}
The algorithm {\sc Approx1} is a $\frac {15}8$-approximation for the {\sc P$4$VC} problem on cubic graphs.
\end{theorem}

\section{{\sc P$4$VC} on $4$-regular graphs}
The design ideas in the above algorithm {\sc Approx1} for cubic graphs do not trivially extend to $4$-regular graphs,
for one of the most important reasons that there are many more configurations for two triangles being adjacent (due to degree $4$) and
the maximum degree of the similarly constructed triangle graph ${\cal G} = (T, R)$ can be as high as $9$.
Such a high maximum degree voids the effectiveness of the best approximation algorithm for {\sc MIS} on degree-$9$ graphs.

We present next an approximation algorithm, denoted as {\sc Approx2}, for {\sc P$4$VC} on $d$-regular graphs when $d \ge 4$ is even,
and show that its performance ratio is $\frac {(3d-2)(2d-2)}{(3d+4)(d-2)}$.
When $d = 4$, the ratio is $\frac {15}8 = 1.875$.
We are able to provide a better but slightly more complex analysis for $d = 4$ to
show that the approximation ratio is actually no greater than $1.852$;
we also use an instance to show that the ratio $1.852$ is almost tight for {\sc Approx2} when $d = 4$.

\subsection{An approximation algorithm when $d \ge 4$ is even}
One of the design ideas in our algorithm is borrowed from Devi et al.~\cite{DMM15}.
Let $G = (V, E)$ denote the input $d$-regular graph, where $d \ge 4$ is even.

In the algorithm {\sc Approx2},
we first compute a subset $V_1$ of vertices by iteratively adding to it a degree-$d$ vertex until no more degree-$d$ exists in the remaining graph;
then similarly and sequentially compute a subset $V_i$ of vertices by iteratively adding to it a degree-$(d - i + 1)$ vertex until
no more degree-$(d - i + 1)$ exists in the remaining graph, for $i = 2, 3, \ldots, d-2$.
Denote
\[
V_{d-1} = V - \cup_{i=1}^{d-2} V_i.
\]
The last remaining graph is $G[V_{d-1}]$, which has maximum degree $2$ and thus an optimal ({\it i.e.}, minimum) $4$-path vertex cover, denoted as $U_{d-1}$,
can be computed in $O(|V_{d-1}|)$ time.

We will prove in Theorem~\ref{thm04} that $U_{d-1} \cup \cup_{i=1}^{d-2} V_i$ is a $4$-path vertex cover of the input graph $G$,
so is $V - (V_{2i-1} \cup V_{2i})$, for each $i = 1, 2, \ldots, \frac d2 - 1$.
The algorithm {\sc Approx2} outputs the smallest among these $\frac d2$ covers as the final solution.
A high-level description of {\sc Approx2} is depicted in Figure~\ref{fig03},
and we prove in Theorem~\ref{thm04} that {\sc Approx2} is an $O(d^2 |V|)$-time $\frac {(3d-2)(2d-2)}{(3d+4)(d-2)}$-approximation
for {\sc P$4$VC} on $d$-regular graphs, when $d \ge 4$ is even.

\begin{figure}[ht]
\begin{center}
\framebox{
\begin{minipage}{5.0in}
	The algorithm {\sc Approx2} for {\sc P$4$VC} on $d$-regular graphs ($d \ge 4$ is even):
	\begin{description}
	\item[Step 1.]
		Compute $V_1$ by iteratively adding to it a degree-$d$ vertex until no more degree-$d$ exists in $G[V - V_1]$;
	\item[Step 2.]
		for $i = 2, 3, \ldots, d-2$,
		similarly and sequentially compute $V_i$ by iteratively adding to it a degree-$(d-i+1)$ vertex until
		no more degree-$(d-i+1)$ exists in $G[V - \cup_{j=1}^i V_j)]$;
	\item[Step 3.]
		set $V_{d-1} = V - \cup_{i=1}^{d-2} V_i$;\\
		find an optimal $4$-path vertex cover $U_{d-1}$ in $G[V_{d-1}]$;
	\item[Step 4.]
		return the smallest one among $U_{d-1} \cup \cup_{i=1}^{d-2} V_i$ and $V - (V_{2i-1} \cup V_{2i})$, $i = 1, 2, \ldots, \frac d2 - 1$.
	\end{description}
\end{minipage}}
\end{center}
\caption{A high-level description of the approximation algorithm {\sc Approx2} for the {\sc P$4$VC} problem on $d$-regular graphs when $d \ge 4$ is even.\label{fig03}}
\end{figure}

\begin{theorem}
\label{thm04}
For the {\sc P$4$VC} problem on $d$-regular graphs, when $d \ge 4$ is even, 
the algorithm {\sc Approx2} is an $O(d^2 |V|)$-time $\frac {(3d-2)(2d-2)}{(3d+4)(d-2)}$-approximation.
\end{theorem}
\begin{proof}
First of all, since $G = (V, E)$ is $d$-regular, we have $|E| = \frac d2 |V|$;
therefore, $V_i$ is computed in $O(d |V|)$ time, for each $i = 1, 2, \ldots, d-2$.
Computing $U_{d-1}$ needs only $O(|V|)$ time since each connected component of $G[V_{d-1}]$ is either a simple cycle or a simple path.
That is, the running time of {\sc Approx2} is in $O(d^2 |V|)$.

Next, we conclude that the vertices of $V_i$, for each $i$, are pairwise non-adjacent to each other.
Also, in the induced subgraph graph $G[V - \cup_{j=1}^{2i-2} V_j]$ which has maximum degree $d-2i+2$,
since every vertex of the computed subset $V_{2i-1}$ has degree exactly $d-2i+2$,
a vertex of $V_{2i}$ can be adjacent to at most one vertex of $V_{2i-1}$.
These two properties suggest that the longest path in the subgraph induced on $V_{2i-1} \cup V_{2i}$, $G[V_{2i-1} \cup V_{2i}]$,
contains at most three vertices (two of $V_{2i}$ and one of $V_{2i-1}$),
that is, $V - (V_{2i-1} \cup V_{2i})$ is a $4$-path vertex cover in $G$, for each $i = 1, 2, \ldots, \frac d2 - 1$.
On the other hand, since $U_{d-1}$ is a $4$-path vertex cover in $G[V_{d-1}]$, $U_{d-1} \cup \cup_{i=1}^{d-2} V_i$ is a $4$-path vertex cover in $G$.
This proves that the solution returned by {\sc Approx2} is feasible.

Recall that each connected component of $G[V_{d-1}]$ is either a simple cycle or a simple path.
The optimal $4$-path vertex cover $U_{d-1}$ in $G[V_{d-1}]$ contains exactly $\lfloor \frac k4\rfloor$ vertices from each $k$-path,
and contains exactly $\lceil \frac k4\rceil$ vertices from each $k$-cycle~\cite{ACB12,BJK13}.
From the fact that
$\left\lfloor \frac k4\right\rfloor \le \frac 14 k$ and $\left\lceil \frac k4\right\rceil \le \frac 25 k$ for all $k \ge 4$,
we conclude that
\[
|U_{d-1}| \le \frac 25 |V_{d-1}|.
\]
%
Consequently, the size of $U_{d-1} \cup \cup_{i=1}^{d-2} V_i$ is
\begin{equation}
\label{eq8}
|U_{d-1}| + \sum_{i=1}^{d-2} |V_i| \le \frac 25 |V| + \frac 35 \sum_{i=1}^{d-2} |V_i|.
\end{equation}
It follows that the minimum cardinality of these $\frac d2$ $4$-path vertex covers is at most
\[
\frac 6{3d+4} \sum_{i=1}^{\frac d2 - 1} \Big(|V| - |V_{2i-1}| - |V_{2i}|\Big)
	+ \frac {10}{3d+4} \left(\frac 25 |V| + \frac 35 \sum_{i=1}^{d/2 - 1} \Big(|V_{2i-1}| + |V_{2i}|\Big)\right)
	= \frac {3d-2}{3d+4} |V|.
\]
From Eq.~(\ref{eq2}) we obtain the lower bound of $\frac {d-2}{2d-2} |V|$ on $\psi_4(G)$ using $k = 4$~\cite{BJK13},
and therefore we prove that the algorithm {\sc Approx2} has an approximation ratio of $\frac {(3d-2)(2d-2)}{(3d+4)(d-2)}$.
Note that due to $\frac {3d-2}{3d+4} < \frac {d/2}{d/2+1}$,
the above ratio is strictly less than that stated in Theorem~\ref{thm02} (or the one by Devi et al.~\cite{DMM15}).
\end{proof}

\begin{corollary}
\label{coro01}
For the {\sc P$4$VC} problem on $4$-regular graphs, the algorithm {\sc Approx2} is an $O(|V|)$-time $\frac {15}8$-approximation.
\end{corollary}

\subsection{{\sc Approx2} is a $1.852$-approximation}
In this section, we present a better analysis for our algorithm {\sc Approx2} than what is done in the proof of Theorem~\ref{thm04}.
Theorem~\ref{thm04} leads to the conclusion in Corollary~\ref{coro01} that {\sc Approx2} is a $\frac {15}8$-approximation for {\sc P$k$VC} on $4$-regular graphs.
Our better analysis shows that the performance ratio of {\sc Approx2} is actually at most $1.852$.

Recall that when $d = 4$, our algorithm {\sc Approx2} (see Figure~\ref{fig03}) computes $V_1$ and $V_2$,
and computes an optimal $4$-path vertex $U_3$ in $G[V_3]$, where $V_3 = V - (V_1 \cup V_2)$.
Both $V_3$ and $U_3 \cup (V_1 \cup V_2)$ are feasible $4$-path vertex covers in the input graph $G$.
{\sc Approx2} returns the smaller one between $V_3$ and $U_3 \cup (V_1 \cup V_2)$, denoted as $A$.
In the following we denote $V_1 \cup V_2$ as $V_{1,2}$.

\subsubsection{An outline of the analysis}
Throughout the analysis, we fix an arbitrary optimal $4$-path vertex cover $B$ in $G$ for discussion.

We color the vertices of $B$ {\em black} and color the other vertices of $G$ (that is, $V - B$) {\em white}.
An edge of $G$ is {\em black} (respectively, {\em white}) if both of its endpoints are black (respectively, white);
an edge of $G$ neither black nor white is {\em bicolor}.
Using this coloring scheme, a bicolor $(V_{1,2}, V_3)$-edge has its endpoint in $V_{1,2}$ white and its endpoint in $V_3$ black.
Bicolor $(V_{1,2}, V_{1,2})$-edges, $(V_3, V_{1,2})$-edges, and $(V_3, V_3)$-edges are defined similarly.

Recall that each connected component of $G[V - B]$ is one of the following five kinds:
\begin{enumerate}
\parskip=0pt
\item[A)]
	a $1$-path,
\item[B)]
	a $2$-path,
\item[C)]
	a $3$-path,
\item[E)]
	a $K_{1,\ell}$ claw, for $\ell = 3, 4$, and
\item[F)]
	a triangle ({\it i.e.}, a $C_3$ cycle).
\end{enumerate}
We merge the first four kinds and name them uniformly a {\em star}.
The {\em center} vertex of a star $S$ is the one with the maximum degree (tie broken arbitrarily),
and all the other vertices (can be $0, 1, 2, 3$, or $4$ of them) are referred to as the {\em satellites} of $S$.
This way, each connected component of $G[V - B]$ is either a triangle or a star.

Our goal is to show that $|A| \le 1.852 |B|$.
Let $F = U_3 \cup V_{1,2}$.
Since $|A| = \min\{|V_3|, |F|\} \le (1 - \alpha) |V_3| + \alpha |F|$ for any coefficient $0 \le \alpha \le 1$,
it suffices to show that there is a constant $0 \le \alpha \le 1$ such that
$(1 - \alpha) |V_3| + \alpha |F| \le 1.852 |B|$.
In the remainder of this section, we show that $\alpha = 0.556$ satisfies this inequality. 
To reach our goal, we first make two important observations summarized in the next two lemmas, respectively.

\begin{lemma}
\label{lemma02}
Let $b_e$ be the number of black edges in $G$, and $s_c$ be the number of star components in $G[V - B]$.
Then, $|B| = \frac 13 |V| + \frac 13 b_e + \frac 13 s_c$. 
\end{lemma}
\begin{proof}
We apply a similar counting as in the proof of Lemma~\ref{lemma01}.
Let $x_i$ denote the number of star components in which the star has $i$ satellites, for $i = 0, 1, 2, 3, 4$;
and $y$ denote the number of triangle components.
It follows that
\[
|V| = |B| + \sum_{i=0}^4 (i+1) x_i + 3 y.
\]
Because $G$ is $4$-regular and there are $b_e$ black edges, the number of bicolor edges (each connecting a vertex of $B$ and a vertex of $V - B$) is exactly
\[
4 |B| - 2 b_e = \sum_{i=0}^4 2 (i+2) x_i + 6 y.
\]
By eliminating $y$ from the above two equalities, we have 
\[
2 |V| - 2 |B| - \sum_{i=0}^4 2 (i+1) x_i = 4 |B| - 2 b_e - \sum_{i=0}^4 2 (i+2) x_i.
\]
Using $s_c = \sum_{i=0}^4 x_i$, we have $3 |B| = |V| + b_e + s_c$.
This proves the lemma.
\end{proof}

\begin{lemma}
\label{lemma03}
Let $p_{3\downarrow}$ (respectively, $p_{4\uparrow}$) be the total number of vertices in those connected components of $G[V_3]$
each is a path of order at most $3$ (respectively, at least $4$).
Let $c_3$ (respectively, $c_{4,6\uparrow}$) be the total number of vertices in those connected components of $G[V_3]$
each is a cycle of order exactly $3$ (respectively, exactly $4$ or at least $6$).
Then, $|U_3| \le \frac 25 |V_3| - \frac 25 \left(p_{3\downarrow} + c_3\right) - \frac 3{20} p_{4\uparrow} - \frac 1{15} c_{4,6\uparrow}$. 
\end{lemma}
\begin{proof}
It is known that $\psi_4(P_\ell) = \lfloor\frac {\ell}4 \rfloor$ and $\psi_4(C_\ell) = \lceil \frac {\ell}4 \rceil$,
where $P_\ell$ and $C_\ell$ are a simple path and a simple cycle of order $\ell$, respectively~\cite{ACB12,BJK13}.
Let $c_5$ be the total number of vertices in those connected components of $G[V_3]$ each is a cycle of order exactly $5$.
Therefore,
\[
|U_3| \le \frac 14 p_{4\uparrow} + \frac 13 c_{4,6\uparrow} + \frac 25 c_5.
\]
Using $|V_3| = p_{3\downarrow} + p_{4\uparrow} + c_3 + c_{4,6\uparrow} + c_5$ to cancel out $c_5$, we achieve the inequality stated in the lemma.
\end{proof}

By the above Lemmas~\ref{lemma02} and \ref{lemma03}, it remains to show the following inequality:
\begin{equation}
\label{eq11}
\frac {\alpha|V_1| + \alpha|V_2| + (1 - \frac {3\alpha}5) |V_3| - \frac {2\alpha}5 \left(p_{3\downarrow} + c_3\right)
	- \frac {3\alpha}{20} p_{4\uparrow} - \frac {\alpha}{15} c_{4,6\uparrow}}
	{\frac 13 |V| + \frac 13 b_e + \frac 13 s_c}
\le 1.852. 
\end{equation}

In $|B| = \frac 13 |V| + \frac 13 b_e + \frac 13 s_c$ (the denominator in Eq.~(\ref{eq11})),
we call $\frac 13 |V|$ the {\em basic lower bound} on $|B|$ and call $\frac 13 b_e + \frac 13 s_c$ the {\em extra lower bound} on $|B|$.

Similarly, in $|A| \le \alpha|V_1| + \alpha|V_2| + (1 - \frac {3\alpha}5) |V_3| - \frac {2\alpha}5 \left(p_{3\downarrow} + c_3\right)
	- \frac {3\alpha}{20} p_{4\uparrow} - \frac {\alpha}{15} c_{4,6\uparrow}$ (the numerator in Eq.~(\ref{eq11})),
we call $\alpha |V_1| + \alpha |V_2| + \left(1 - \frac {3\alpha}5 \right) |V_3|$ the {\em basic upper bound} on $|A|$ and
	call $\frac {2\alpha}5 \left(p_{3\downarrow} + c_3\right) + \frac {3\alpha}{20} p_{4\uparrow} + \frac {\alpha}{15} c_{4,6\uparrow}$ the {\em saving} on $|A|$. 

Roughly speaking, we used only the basic lower bound on $|B|$ (as in Eq.~(\ref{eq2}) and
	the basic upper bound $|A|$ (with $\alpha = 1$, as in Eq.~(\ref{eq8})) in the proof of Theorem~\ref{thm04} (when $d = 4$).
In other words, Lemma~\ref{lemma02} gives a better lower than Eq.~(\ref{eq2}) when $k = 4$ and $d = 4$,
and Lemma~\ref{lemma03} gives a better estimation than Eq.~(\ref{eq8}) when $d = 4$.
The extra lower bound and the saving will help us get a better analysis.

It seems difficult to verify Eq.~(\ref{eq11}) if we consider the graph $G$ as a whole. 
So, to ease the proof of Eq.~(\ref{eq11}), we consider the following two kinds of subgraphs of $G$ and want to verify Eq.~(\ref{eq11}) on each of these subgraphs. 
\begin{itemize}
\item {\em Type-1:}
	The subgraphs of this type one-to-one correspond to the connected components of $G[V - B]$, and they are constructed as follows.
	Consider a connected component $K$ of $G[V - B]$, which is either a $C_3$ cycle or a star, and all its vertices are white.
	Let $K'$ be the subgraph of $G$ induced by the vertices of $K$ and their black neighbors in $G$.
	Let $H$ be the graph obtained from $K'$ by deleting all black edges.
	Then, $H$ is the type-1 subgraph corresponding to $K$. 
\item {\em Type-2:}
	The subgraphs of this type one-to-one correspond to the black edges in $G$.
	That is, the subgraph corresponding to a black edge $e$ consists of only $e$ and its two ending black vertices. 
\end{itemize}
Let ${\cal G}_1$ (respectively, ${\cal G}_2$) be the collection of type-1 (respectively, type-2) subgraphs in $G$.
Obviously, each white vertex of $G$ appears in exactly one subgraph in ${\cal G}_1 \cup {\cal G}_2$.
In contrast, a black vertex of $G$ can appear in one or more subgraphs in ${\cal G}_1 \cup {\cal G}_2$.
Nevertheless, each edge of $G$ appears in exactly one subgraph in ${\cal G}_1 \cup {\cal G}_2$.

To prove Eq.~(\ref{eq11}), we proceed as follows:
\begin{description}
\item[{\it Step 1:}]\label{step:1} 
	Distribute the numerator and the denominator of the left hand side of the inequality to the subgraphs in ${\cal G}_1\cup {\cal G}_2$. 
\item[{\it Step 2:}]\label{step:2} 
	Prove that for each subgraph $H$ in ${\cal G}_1\cup {\cal G}_2$, $\frac {n_H}{d_H} \le 1.852$,
	where $n_H$ (respectively, $d_H$) is the portion of the numerator (respectively, denominator) of the left hand side of the inequality distributed to $H$. 
\end{description}

The next three subsections are devoted to detailing the above two steps, respectively.
After these two steps, we are done because for any sequence of positive numbers $x_1, x_2, \ldots, x_k, y_1, y_2, \ldots, y_k$,
$\frac{x_i}{y_i} \le 1.852$ for all $i = 1, 2, \ldots, k$ implies $\frac {\sum_{i=1}^k x_i}{\sum_{i=1}^k y_i} \le 1.852$.

\subsection{Distributing the denominator}
Initially, we distribute the basic lower bound (namely, $\frac 13 |V|$) evenly to the edges in $G$ so that each edge holds a basic lower bound of $\frac 16$;
we further distribute the extra lower bound (namely, $\frac 13 b_e + \frac 13 s_c$) evenly to the black edges in $G$ and the star components of $G[V- B]$
so that each black edge holds an extra lower bound of $\frac 13$ and so does each star component of $G[V- B]$.

If a $5$-cycle $C_5$ in $G[V_3]$ has at least one black edge, then it is {\em good};
otherwise, it is {\em bad}.
Consider a bad $5$-cycle $C_5$ in $G[V_3]$.
Since $B$ is a solution ({\it i.e.}, $4$-path vertex cover), $C_5$ must have exactly two black vertices and a unique white edge.
Let $e$ be the white edge in $C_5$, and $v$ be the white vertex of $C_5$ that is not an endpoint of $e$.
We call $v$ the {\em independent white vertex} in $C_5$.
If $v$ or $e$ appears in a star component of $G[V- B]$ or at least one vertex of $C_5$ is incident to a black edge in $G$, 
then $C_5$ is {\em slightly bad};
otherwise, $C_5$ is {\em very bad}. 
A simple but important observation is that no star component of $G[V- B]$ can contain both $v$ and $e$.
A $(V_{1,2}, V_3)$-edge of $G$ is {\em good} if its black endpoint either is an endpoint of a black edge in $G[V_3]$ 
or appears in a good or slightly bad $5$-cycle of $G[V_3]$.

First, consider a connected component $K$ of $G[V_3]$ that has at least one black edge.
Let $p$ be the number of black edges in $K$, and $q$ be the number of good $(V_{1,2}, V_3)$-edges in $G$ whose black endpoints appear in $K$.
We collect the extra lower bounds held by the black edges of $K$;
the total is obviously $\frac 13 p$.
From this total, we distribute $0.015q$ evenly to the $q$ good $(V_{1,2}, V_3)$-edges so that each of them receives $0.015$,
and then distribute the remaining (namely, $\frac 13 p - 0.015 q$) to the $p$ black edge so that each of them receives $\frac 13 - \frac {0.015 q}p$.
Since $\frac qp \le 6$, each black edge in $K$ still holds an extra lower bound of $\ge \frac {73}{300}$.

Next, consider a slightly bad $5$-cycle $C_5$ in $G[V_3]$.
Let $e$ be the white edge in $C_5$, and $v$ be the independent white vertex of $C_5$.
We transfer a portion of the extra lower bound (namely, $\frac 13$) held by $C_5$ as follows (three possible cases): 
\begin{itemize}
\item
	Suppose that $e$ appears in a star component $K$ of $G[V- B]$.
	Then, we say that $C_5$ is of {\em type-1}.
	Among the extra lower bound (namely, $\frac 13$) held by $K$, we transfer $0.015$ to each good $(V_{1,2}, V_3)$-edge whose black endpoint appears in $C_5$.
	Obviously, $C_5$ is the unique bad $5$-cycle whose white edge appears in $K$.
	Moreover, there are exactly $4$ good $(V_{1,2}, V_3)$-edges whose black endpoints appear in $C_5$.
	Thus, the extra lower bound still held by $K$ is $\frac 13 - 4 \times 0.015 = \frac {41}{150}$. 
\item
	Suppose that $v$ appears in a star component $K$ of $G[V- B]$ but $e$ does not.
	Then, we say that $C_5$ is of {\em type-2}.
	Among the extra lower bound (namely, $\frac 13$) held by $K$, we transfer $0.015$ to each good $(V_{1,2}, V_3)$-edge whose black endpoint appears in $C_5$.
	Obviously, there are at most $4$ bad $5$-cycles whose independent white vertices appear in $K$.
	Moreover, each bad $5$-cycle contains the black endpoints of exactly $4$ good $(V_{1,2}, V_3)$-edges.
	Thus, the extra lower bound still held by $K$ is $\frac 13 - 4 \times 4 \times 0.015 = \frac 7{75}$. 
\item
	Suppose that at least one vertex of $C_5$ is incident to a black edge in $G$. 
	Then, we say that $C_5$ is of {\em type-3}.
	We choose an arbitrary black edge $g$ such that one endpoint of $g$ is in $C_5$.
	Among the extra lower bound (namely, $\frac 13$) held by $g$, we transfer $0.015$ to each good $(V_{1,2}, V_3)$-edge whose black endpoint appears in $C_5$.
	Obviously, there are at most $3$ good $(V_{1,2}, V_3)$-edges whose black endpoints appear in $C_5$.
	Thus, the extra lower bound still held by $g$ is at least $\frac 13 - 3 \times 0.015 = \frac {173}{600}$. 
\end{itemize}
Table~\ref{tab04} summarizes how the basic and the extra lower bounds have been distributed to the edges of $G$ and the star components of $G[V - B]$. 

\begin{table}[htbp]
\caption{Fractional lower bounds distributed to the edges of $G$ and the star components of $G[V - B]$.\label{tab04}}
\begin{center}
\begin{tabular}{l||c|c|c|c}
From lower bound	& basic: $\frac 13 |V|$	& \multicolumn{3}{c}{extra: $\frac 13 b_e + \frac 13 s_c$} \\
\hline
\hline
Entity				& an edge 		& a black edge 			& a good $(V_{1,2}, V_3)$-edge 	& a star component \\
\hline
Fraction received	& $\frac 13$ 	& $\ge \frac {73}{300}$ & $0.015$ 						& $\ge \frac 7{75}$\\
\end{tabular}
\end{center}
\end{table}

\subsection{Distributing the numerator}
Initially, we distribute the basic upper bound on $A$ to the vertices of $G$ so that each $v \in V_{1,2}$ receives $\alpha$
while each $v \in V_3$ receives $1 - \frac {3\alpha}5$. 
Moreover, we distribute the saving on $|A|$ to the vertices in $V_3$ so that
each vertex in a $3$-cycle or a path of order at most $3$ in $G[V_3]$ receives $\frac  25 \alpha$, 
each vertex in a path of order at least $4$ vertices in $G[V_3]$ receives $\frac 3{20} \alpha$,
and each vertex in a cycle of order exactly $4$ or at least $6$ in $G[V_3]$ receives $\frac 1{15} \alpha$.

Since a black vertex of $G$ may belong to two or more subgraphs in ${\cal G}_1 \cup {\cal G}_2$, 
we next distribute the basic upper bound (namely, $\alpha$) held by each black $v \in V_1\cup V_2$ evenly to the edges incident to $v$ in $G$.
For the same reason, we distribute the basic upper bound (namely, $1 - \frac {3\alpha}5$) held by each black vertex $v \in V_3$
evenly to the edges incident to $v$ in $G$.
Table~\ref{tab05} summarizes how the basic upper bounds have been distributed. 

\begin{table}[htbp]
\caption{Fractional basic upper bounds distributed to the entities in $G$.\label{tab05}}
\begin{center}
\begin{tabular}{l||c|c|c|c|c|c|c|c}
 & \multicolumn{8}{c}{From basic upper bound: $\alpha |V_1| + \alpha |V_2| + \left(1 - \frac {3\alpha}5 \right) |V_3|$} \\
\hline
\hline
Entity	& \multicolumn{2}{c|}{a white vertex} 	& a black vertex 	& \multicolumn{2}{c|}{a bicolor edge} 	& \multicolumn{3}{c}{a black edge} \\
\cline{2-3} \cline{5-9}
		& in & in & or & black end & black end & both ends & both ends & others\\
		& $V_{1,2}$ & $V_3$ & a white edge & in $V_{1,2}$ & in $V_3$ & in $V_{1,2}$ & in $V_3$ \\
\hline
Received & $\alpha$ & $\frac{5 - 3\alpha}5$ & 0 & $\frac 14\alpha$ & $\frac{5-3\alpha}{20}$ & $\frac 12\alpha$ & $\frac{5-3\alpha}{10}$ & $\frac{5+2\alpha}{20}$ \\
\end{tabular}
\end{center}
\end{table}

We further consider transferring the savings held by the vertices in $V_3$. 
In some cases, we might discard all or a portion of the saving held by a vertex. 
In other words, we will show that Eq.~(\ref{eq11}) still holds even if we discard a portion of the saving on $|A|$%
\footnote{This suggests our final ratio $1.852$ is not strictly tight, but only almost tight.}

A $(V_{1,2}, V_3)$-edge of $G$ is {\em saving} if
its black endpoint not only has no black neighbor in $G[V_3]$ but also appears in a connected component of $G[V_3]$ that is not a $5$-cycle.
Consider such a connected component $K$ of $G[V_3]$ (that is, $K \ne C_5$).
For each white vertex $v$ of $K$, we perform the following:
\begin{itemize}
\item
	If $v$ has no black neighbor in $K$ (case 1), then we discard the whole saving held by $v$;
	otherwise (case 2), we distribute the whole saving held by $v$ evenly to its neighbors in $K$. 
\end{itemize}
After we have done the above for all the white vertices of $K$,
for each black vertex $u$ of $K$, if $u$ is incident to no saving $(V_{1,2}, V_3)$-edge in $G$, 
then we discard the whole saving held by $u$;
otherwise, we perform the following: 
\begin{itemize}
\item
	When $K$ is a cycle, $u$ is incident to one or two saving $(V_{1,2}, V_3)$-edges in $G$ and $u$ now holds a saving of at least 
	$\frac 1{15} \alpha + 2 \times \frac 1{30} \alpha = \frac 2{15} \alpha$, 
	in which the original saving held by $u$ was at least $ \frac 1{15} \alpha$ and
	$u$ received at least $\frac 1{30} \alpha$ from each of its white neighbors in $K$. 
	We distribute the whole saving held by $u$ evenly to the saving $(V_{1,2}, V_3)$-edges 
	incident to $u$ in $G$.
	Note that each saving $(V_{1,2}, V_3)$-edge incident to $u$ receives a saving of at least $\frac 1{15} \alpha$.
\item
	When $K$ is a singleton $u$, $u$ is incident to at most four saving $(V_{1,2}, V_3)$-edges in $G$ and the saving held by $u$ is $\frac 2{5} \alpha$. 
	We distribute the whole saving held by $u$ evenly to the saving $(V_{1,2}, V_3)$-edges incident to $u$ in $G$.
	Each saving $(V_{1,2}, V_3)$-edge incident to $u$ receives a saving of at least $\frac 1{10} \alpha$. 
\item
	When $K$ is a path of order at least $2$, $u$ is incident to at most three saving $(V_{1,2}, V_3)$-edges in $G$ and $u$ now holds a saving of
	at least $\frac 3{20} \alpha + \frac 3{40} \alpha = \frac 9{40} \alpha$,
	in which the original saving held by $u$ was at least $\frac 3{20} \alpha$ and
	$u$ received at least $\frac 3{40} \alpha$ from its white neighbors in $K$.
	We distribute the whole saving held by $u$ evenly to its incident saving $(V_{1,2}, V_3)$-edges.
	Each saving $(V_{1,2}, V_3)$-edge incident to $u$ receives a saving of at least $\frac 3{40} \alpha$.
\end{itemize}
Table~\ref{tab06} summarizes how the savings have been distributed.

\begin{table}[htbp]
\caption{Savings distributed to the entities in $G$.}\label{tab06}
\begin{center}
\begin{tabular}{l||c|c|c}
 & \multicolumn{3}{c}{From saving: $\frac {2\alpha}5 \left(p_{3\downarrow} + c_3\right) + \frac {3\alpha}{20} p_{4\uparrow} + \frac {\alpha}{15} c_{4,6\uparrow}$} \\
\hline
\hline
Entity		& saving $(V_{1,2}, V_3)$-edge 	& any other edge 	& any vertex \\
\hline
Received	& $\ge\frac 1{15} \alpha$ 		& 0 				& 0 \\
\end{tabular}
\end{center}
\end{table}

\subsection{Verifying Eq.~(\ref{eq11}) on subgraphs}
For a subgraph $H$ of $G$, we define the following notations:
\begin{itemize}
\item
	$u_H$: the total basic upper bound distributed to the vertices and edges of $H$;
\item
	$s_H$: the total saving distributed to the saving $(V_{1,2}, V_3)$-edges of $H$;
\item
	$\ell^b_H$: the total basic lower bound distributed to the edges of $H$; 
\item
	$\ell^e_H$: the total extra lower bound distributed to subgraphs of $H$. 
\end{itemize}
Our goal is to show that, for each $H \in {\cal G}_1\cup {\cal G}_2$, $\frac {u_H - s_H}{\ell^b_H + \ell^e_H} \le 1.852$.

Unfortunately, we may fail to do this for some subgraphs $H \in {\cal G}_1$, because of very bad $5$-cycles.
We call a $(V_{1,2}, V_3)$-edge {\em very bad} if its black endpoint appears in a very bad $5$-cycle in $G[V_3]$. 
Compared with other bicolor edges, very bad $(V_{1,2},V_3)$-edges hold large basic upper bounds but neither savings nor extra lower bounds. 

For each very bad $5$-cycle $C$ in $G[V_3]$,
the white edge of $C$ appears in a $3$-cycle in $G[V- B]$ while the independent white edge of $C$ appears in another $3$-cycle $C'$ in $G[V- B]$.
We call $C'$ the {\em companion $3$-cycle} of $C$ (see Figure~\ref{fig04}(a)).
\begin{figure}[htpb]
\begin{center}
\begin{subfigure}{0.49\textwidth}
  \setlength{\unitlength}{0.7bp}%
  \begin{picture}(298.15, 87.77)(0,0)
  \put(0,10){\includegraphics[scale=0.7]{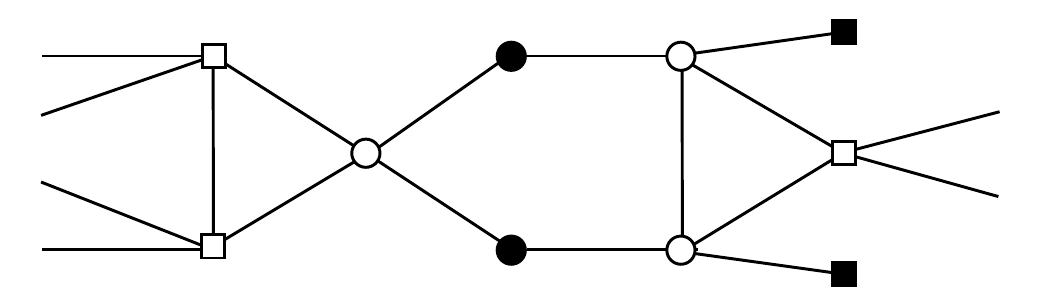}}
  \put(142.87,40.67){\fontsize{11.38}{13.66}\selectfont $C$}
  \put(64.70,42.89){\fontsize{11.38}{13.66}\selectfont $C'$}
  \put(6.29,75.31){\fontsize{8.54}{10.24}\selectfont ?}
  \put(5.77,57.74){\fontsize{8.54}{10.24}\selectfont ?}
  \put(5.67,38.37){\fontsize{8.54}{10.24}\selectfont ?}
  \put(6.57,19.20){\fontsize{8.54}{10.24}\selectfont ?}
  \put(288.63,58.29){\fontsize{8.54}{10.24}\selectfont ?}
  \put(288.70,33.84){\fontsize{8.54}{10.24}\selectfont ?}
  \end{picture}%
\caption{\label{fig04a}}
\end{subfigure}
\begin{subfigure}{0.49\textwidth}
  \setlength{\unitlength}{0.7bp}%
  \begin{picture}(262.43, 97.82)(0,0)
  \put(0,10){\includegraphics[scale=0.7]{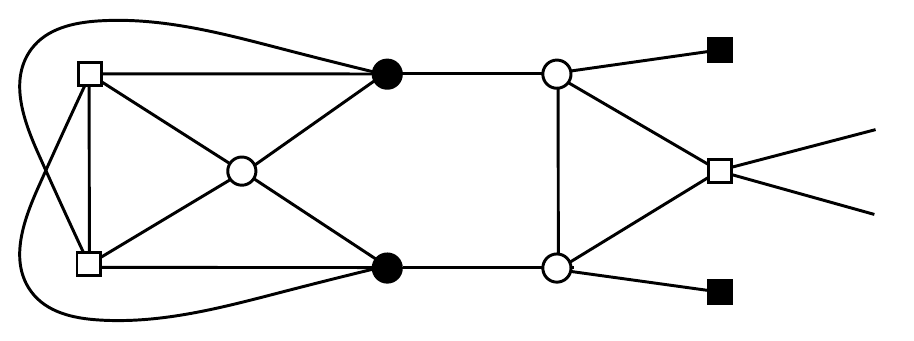}}
  \put(107.16,45.52){\fontsize{11.38}{13.66}\selectfont $C$}
  \put(28.98,47.07){\fontsize{11.38}{13.66}\selectfont $C'$}
  \put(252.92,62.48){\fontsize{8.54}{10.24}\selectfont ?}
  \put(252.99,38.02){\fontsize{8.54}{10.24}\selectfont ?}
  \end{picture}%
\caption{\label{fig04b}}
\end{subfigure}
\end{center}
\caption{A very bad $5$-cycle $C$ and its companion $3$-cycle $C'$, where circle vertices are in $V_3$, square vertices are in $V_{1,2}$,
	and each question mark is a black circle or a black square.\label{fig04}}
\end{figure}
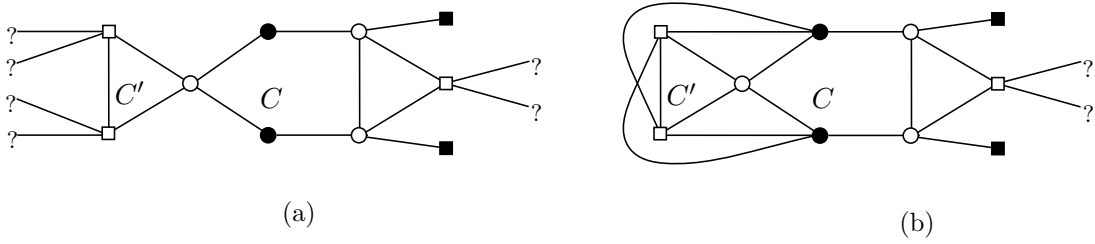
From Figure~\ref{fig04}(a), we can see that each very bad $5$-cycle can contain the black endpoints of at most four very bad $(V_{1,2},V_3)$-edges in $G$;
we also see that $C'$ can accommodate the white endpoints of up to four very bad $(V_{1,2},V_3)$-edges in $G$.
Thus, $|X| \le |Y|$, where $X$ is the set of very bad $(V_{1,2},V_3)$-edges in $G$ whose black endpoints are in $C$ but whose white endpoints are not in $C'$,
and $Y$ is the set of $(V_{1,2},V_3)$- or $(V_{1,2},V_{1,2})$-edges in $G$ whose white endpoints are in $C'$ but whose black endpoints are not in $C$.
So, we can find an injection $f$ from $X$ to $Y$. 
As long as we are only concerned with the proof of Eq.~(\ref{eq11}), it is no harm to exchange the black endpoint of each $x \in X$ with the black endpoint of $f(x)$,
because the exchange does not change the total basic or extra lower bound, the total basic upper bound, or the total saving held by $x$ and $f(x)$.  
Therefore, we may assume without loss of generality that:

\begin{itemize}
\item
	(Assumption 1.)
	For every very bad $5$-cycle $C$ in $G[V_3]$, each black vertex of $C$ is adjacent to the two white vertices of the companion 3-cycle $C'$
	(see Figure~\ref{fig04}(b)). 
\end{itemize}

Now, we are ready to show that $\frac {u_H - s_H}{\ell^b_H + \ell^e_H} \le 1.852$, for each $H \in {\cal G}_1\cup {\cal G}_2$.
We remark that all these four quantities $u_H, s_H, \ell^b_H, \ell^e_H \ge 0$ are non-negative, for every $H \in {\cal G}_1$.

First, consider an $H \in {\cal G}_2$, which is a black edge. 
From Table~\ref{tab04}, $\ell^b_H = \frac 16$ and $\ell^e_H \ge \frac {73}{300}$. 
From Table~\ref{tab05}, $u_H \le 2 \times \frac {5-3\alpha}{10} = \frac {5-3\alpha}5$ 
because $\alpha = 0.556 < \frac{5}{8}$.
It follows that $\frac {u_H - s_H}{\ell^b_H + \ell^e_H} \le 1.852$.

Next, consider an $H \in {\cal G}_1$ corresponding to a star component $K$ of $G[V- B]$. 
Note that $K$ can have up to 4 satellites and both $V_{1,2}$ and $V_3$ may contain zero or more vertices of $K$.
Since $\alpha = 0.556$, a tedious but easy analysis shows that $\frac {u_H - s_H}{\ell^b_H + \ell^e_H}$ is maximized when 
\begin{itemize}
\item
	$K$ is a star with $5$ vertices, the center of $K$ is in $V_1 \cup V_3$, and each satellite of $K$ appears in a bad $5$-cycle of $G[V_3]$. 
\end{itemize}
In this worst case, from Table~\ref{tab05}, $u_H \le \alpha + 4\times \frac {5-3\alpha}5 + 4\times \left( 2\times \frac{5-3\alpha}{20} + \frac 14 \alpha\right)
	= \frac {30-8\alpha}5$.
Moreover, from Table~\ref{tab04}, $\ell^b_H = 16\times \frac 13 = \frac 83$, and $\ell^e_H \ge \frac 7{75}$. 
Thus, $\frac{u_H - s_H}{\ell^b_H + \ell^e_H} \le 1.852$.

Next consider an $H \in {\cal G}_1$ corresponding to a $3$-cycle $K$ of $G[V- B]$. 
Note that $H$ has $6$ black vertices.
We distinguish the following three cases:

\begin{description}
\item[{\rm Case 1:}]
	All vertices of $K$ belong to $V_3$. In this case, all black vertices of $H$ belong to $V_{1,2}$.
	Thus, $u_H \le 3\times \frac {5-3\alpha}5 + 6\times \frac 14 \alpha = \frac {30-3\alpha}{10}$ from Table~\ref{tab05},
	$s_H \ge 3\times \frac 1{15} \alpha = \frac 15 \alpha$ from Table~\ref{tab06},
	and $\ell^b_H = 9\times \frac 16 = \frac 32$ from Table~\ref{tab04}. 
	So, $\frac{u_H - s_H}{\ell^b_H + \ell^e_H} \le 1.852$. 
\item[{\rm Case 2:}]
	Exactly one vertex of $K$ belongs to $V_3$.
	In this case, both $V_{1,2}$ and $V_3$ may contain some black vertices of $H$. 
	Since $\alpha = 0.556$, a tedious but easy analysis shows that $\frac{u_H - s_H}{\ell^b_H + \ell^e_H}$ is maximized
	when all black vertices of $H$ belong to $V_3$.
	In this worst case,
	$u_H \le \frac {5-3\alpha}5 + 2\alpha + 6\times \frac{5-3\alpha}{20} = \frac{5+\alpha}2$ from Table~\ref{tab05},
	and $\ell^b_H = 9\times \frac 16$ from Table~\ref{tab04}. 
	So, $\frac{u_H - s_H}{\ell^b_H + \ell^e_H} \le 1.852$. 
\item[{\rm Case 3:}]
	Exactly one vertex $u$ of $K$ belongs to $V_{1,2}$.
	In this case, at least two black vertices of $H$ belong to $V_{1,2}$.
	Moreover, by Assumption~1, each bicolor edge incident to $u$ in $H$ is either a $(V_{1,2}, V_{1,2})$-edge or a good or saving $(V_{1,2}, V_3)$-edge in $K$.
	Since $\alpha = 0.556$, a tedious but easy analysis shows that $\frac{u_H - s_H}{\ell^b_H + \ell^e_H}$ is maximized
	when exactly four black vertices of $H$ belong to $V_{1,2}$.
	In this worst case,
	$u_H \le 2\times \frac{5-3\alpha}5 + \alpha + 2\times \frac{5-3\alpha}{20} + 4\times \frac 14 \alpha = \frac {10+\alpha}5$ from Table~\ref{tab05},
	while $\ell^b_H = 9\times \frac 16$ from  Table~\ref{tab04}.
	So, $\frac{u_H - s_H}{\ell^b_H + \ell^e_H} \le 1.852$. 
\end{description}
In summary, we have proved the following theorem.

\begin{theorem}
\label{thm05}
For the {\sc P$k$VC} problem on $4$-regular graphs, the algorithm {\sc Approx2} actually is a $1.852$-approximation.
\end{theorem}

Table~\ref{tab03} summarizes the improvement we have made for approximating the {\sc P$k$VC} problem on regular graphs.
Comparing against Table~\ref{tab01},
all previously shown approximation ratios have been reduced except those for {\sc P$4$VC} on $d$-regular graphs where $d \ge 5$ is odd.

\begin{table}[ht]
\caption{The improved approximation ratios for {\sc P$k$VC}, $k \ge 3$, on regular graphs.\label{tab03}}
\begin{center}
\begin{tabular}{c||c|c|c|c}
{\sc P$k$VC}	& Cubic 					& $4$-regular
	 					& $d$-regular ($d \ge 5$ is odd)	& $d$-regular ($d \ge 6$ is even)\\
\hline
\hline
$k=3$ 			
						& \multicolumn{4}{c}{$\frac {\lfloor d/2 \rfloor (2d - 1)}{(\lfloor d/2 \rfloor + 1) (d - 1)}$~(Thm~\ref{thm02})}\\[4pt]
\hline
$k=4$ 			& $\frac {15}8$~(Thm~\ref{thm03})	& $1.852$~(Thm~\ref{thm05})
						& $\frac {\lfloor d/2\rfloor (2d - 2)}{(\lfloor d/2\rfloor + 1) (d - 2)}$~\cite{DMM15}
						& $\frac {(3d-2)(2d - 2)}{(3d+4)(d - 2)}$~(Thm~\ref{thm04})\\[4pt]
\hline
$k \ge 5$ 		& n.a.			& \multicolumn{3}{c}{$\frac {\lfloor d/2 \rfloor (2d - k + 2)}{(\lfloor d/2 \rfloor + 1) (d - k + 2)}$, when $k-2 < d$~(Thm~\ref{thm02})}\\
\end{tabular}
\end{center}
\end{table}

\subsection{Almost tight instances}
The analysis of our algorithm that leads to Theorem~\ref{thm05} gives us a constant $\alpha = 0.556$ such that
$\frac {(1-\alpha) |V_3| + \alpha |V_{1,2} \cup U_3|}{|B|} \le 1.852$.
One may wonder whether we can change the value of $\alpha$ so that a smaller upper bound can be achieved.
To answer this question, let us consider the two example graphs $G_1$ and $G_2$ in Figure~\ref{fig05}.

\begin{center}
\begin{figure}[htpb]
\begin{center}
\includegraphics{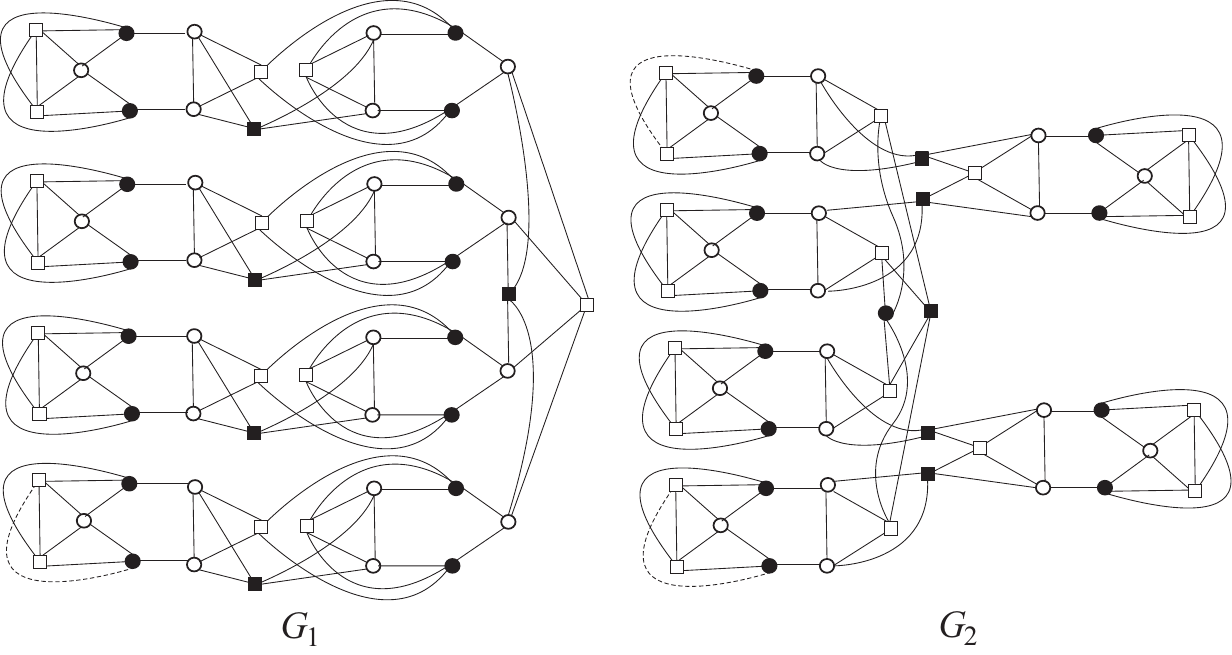}
\end{center}
\caption{Two example graphs for showing the ratio $1.852$ is almost tight for our algorithm {\sc Approx2}.
	In each graph, the circle vertices are in $V_3$, the square vertices are in $V_{1,2}$, respectively,
	and the filled (black) vertices are in the optimal $4$-path vertex cover $B$.\label{fig05}}
\end{figure}
\end{center}

For $G_1$, $|B| = 21$, $|V_3| = 40$, and $|V_{1,2} \cup U_3| = 38$;
for $G_2$, $|B| = 18$, $|V_3| = 31$, and $|V_{1,2} \cup U_3| = 35$. 
Our algorithm outputs $V_{1,2} \cup U_3$ for $G_1$ and outputs $V_3$ for $G_2$.
Solving the equation $\frac {40(1-\alpha) + 38\alpha}{21} = \frac {31(1-\alpha) + 35\alpha}{18}$ yields $\alpha = \frac {23}{40} = 0.575$,
we thus cannot get a better upper bound on $\frac {(1-\alpha) |V_3| + \alpha |V_{1,2} \cup U_3|}{|B|}$ than 
$\frac {31\times (1-\frac{23}{40}) + 35\times \frac{23}{40}}{18} = 1.85$.
Changing the value of $\alpha$ could possibly achieve a smaller upper bound than $1.852$, but only marginally.


In fact, using the above $G_1$ and $G_2$, we can show that $1.85$ is a lower bound on the worst-case performance ratio of our algorithm {\sc Approx2}.
We make a copy of $G_1$, denoted as $G_1^c$, and create a new graph $G$ by letting $G_1$, $G_1^c$, and $G_2$ be its three connected components;
or furthermore we can modify $G$ to be connected by exchanging the black endpoint of the dashed edge in $G_1$ (respectively, in $G_1^c$)
with that of a dashed edge in $G_2$ (see Figure~\ref{fig05}).
Either way, for the new graph $G$ we have $|B| = 21 \times 2 + 18 = 60$, $|V_3| = 40 \times 2 + 31 = 111$, and  $|V_{1,2} \cup U_3| = 38 \times 2 + 35 = 111$.
Therefore, the approximation ratio achieved by our algorithm {\sc Approx2} on $G$ is $\frac {111}{60} = 1.85$.

\section{{\sc P$4$VC} on regular bipartite graphs}
Devi et al.~\cite{DMM15} have shown that {\sc P$4$VC} is APX-complete on cubic bipartite graphs.
On the positive side, the $2$-approximation for {\sc P$3$VC} on (any) bipartite graphs by Kumar et al.~\cite{KMD14} is
also a $2$-approximation for {\sc P$4$VC} on (any) bipartite graphs.
In this section, we present a $\frac {d^2}{d^2 - d + 1}$-approximation for {\sc P$4$VC} on $d$-regular bipartite graphs, for any $d \ge 3$.
Note that this ratio is (much) smaller than the corresponding ratio for {\sc P$4$VC} on $d$-regular graphs as listed in Table~\ref{tab03},
and it approaches $1$ with increasing $d$.

Recall that Eq.~(\ref{eq4}) provides a lower bound on $\psi_4(G)$ when $G$ is $d$-regular,
using the total number $f$ of vertices in all the triangle components of the induced subgraph $G[V - F^*]$ where $F^*$ is an optimal $4$-path vertex cover.
When $G$ is also bipartite, there is no triangle in $G$, that is, $f = 0$.
Thus, Eq.~(\ref{eq4}) becomes
\begin{equation}
\label{eq9}
\psi_4(G) \ge \frac {d-1}{2d} |V|.
\end{equation}

Let $A$ and $B$ denote the two parts in the graph, {\it i.e.}, $G = (A, B, E)$ is $d$-regular bipartite,
where each edge of $E$ connects a vertex of $A$ and a vertex of $B$.
One clearly sees that $|A| = |B| = \frac 12 |V|$.

In the following algorithm for {\sc P$4$VC}, denoted as {\sc Approx3},
we compute a subset $A_1 \subset A$ by iteratively adding to it a degree-$d$ vertex in $A$,
followed by removing the vertex and all its $d$ neighbors (in $B$) from the graph.
The algorithm terminates when there is no degree-$d$ vertex in $A$ in the remaining graph, and returns $A - A_1$ as the solution.

\begin{theorem}
\label{thm06}
The algorithm {\sc Approx3} is an $O(d^2 |V|)$-time $\frac {d^2}{d^2 - d + 1}$-approximation for {\sc P$4$VC} on $d$-regular bipartite graphs.
\end{theorem}
\begin{proof}
Each iteration of adding a degree-$d$ vertex in $A$ to $A_1$ and removing the vertex and all its $d$ neighbors from the graph
takes only $O(d^2)$ time and there are at most $|V|$ iterations.
Therefore, the overall running time of {\sc Approx3} is $O(d^2 |V|)$.

Note that $G$ is $d$-regular and bipartite.
In the subgraph $G[A_1 \cup B]$ induced on $A_1 \cup B$,
each connected component is a $K_{1, d}$ claw if it contains a vertex of $A_1$, or otherwise it is a $1$-path (singleton).
That is, $G[A_1 \cup B]$ does not contain any $4$-path and thus $A - A_1$ is a $4$-path vertex cover;
and in $G[A_1 \cup B]$, there are exactly $|A_1|$ $K_{1, d}$ claws and exactly $|B| - |A_1| d$ singletons.

Since each vertex of $A - A_1$ can be adjacent to at most $d-1$ $1$-path components of $G[A_1 \cup B]$, we have
\[
(|A| - |A_1|) (d-1) \ge (|B| - |A_1| d) d = (|A| - |A_1| d) d.
\]
It follows that
\begin{equation}
\label{eq10}
|A| - |A_1| \le \frac {d^2 - d}{d^2 - d + 1} |A| = \frac {d^2 - d}{2 (d^2 - d + 1)} |V|.
\end{equation}

Using the lower bound in Eq.~(\ref{eq9}), the above Eq.~(\ref{eq10}) leads to
\[
\frac {|A| - |A_1|}{\psi_4(G)} \le \frac {d^2 - d}{2 (d^2 - d + 1)} \times \frac {2d}{d - 1} = \frac {d^2}{d^2 - d + 1}.
\]
This proves the theorem.
\end{proof}

\section{Conclusions and remarks}
In this paper we considered the minimum $k$-path vertex cover problem ({\sc P$k$VC}), which generalizes the classic vertex cover ({\sc VC}) problem.
We designed several improved approximation algorithms,
including a $\frac {\lfloor d/2\rfloor (2d - k + 2)}{(\lfloor d/2\rfloor + 1) (d - k + 2)}$-approximation for {\sc P$k$VC}
on $d$-regular graphs when $1 \le k - 2 < d$.
This beats all the best known approximation results for {\sc P$k$VC} on $d$-regular graphs for $k \ge 3$,
except for {\sc P$4$VC} it ties with the best prior work and in particular they tie at $2$ on cubic graphs and $4$-regular graphs.
We then designed a better $1.875$-approximation for {\sc P$4$VC} on cubic graphs,
a better $1.852$-approximation for {\sc P$4$VC} on $4$-regular graphs,
and a better $\frac {(3d - 2)(2d - 2)}{(3d + 4)(d - 2)}$-approximation for {\sc P$4$VC} on $d$-regular graphs when $d \ge 4$ is even.
We also presented a $\frac {d^2}{d^2 - d + 1}$-approximation for {\sc P$4$VC} on $d$-regular bipartite graphs.

We thus improved all the best prior approximation algorithms for {\sc P$k$VC} on $d$-regular graphs for $k \ge 3$,
except for {\sc P$4$VC} on $d$-regular graphs when $d \ge 5$ is odd, where we tie with the best prior work.
We leave it open on how to take advantage of the odd degrees.

\subsubsection*{Acknowledgements.}
AZ and YC are supported by the NSFC Grants 11771114 and 11571252;
YC is also supported by the CSC Grant 201508330054.
ZZC is supported by in part by the Grant-in-Aid for Scientific Research of the Ministry of Education, Science, Sports and Culture of Japan, under Grant No. 18K11183.
GL is supported by the NSERC Canada.

\end{document}